\documentclass [12pt,a4paper]{article}

\usepackage{pslatex}
\usepackage{amsfonts,color,morefloats}
\usepackage{amssymb,amsmath,latexsym,amsthm}
\usepackage{bbm}
\usepackage{amsfonts}
\usepackage{mathrsfs}
\usepackage{amsthm}
\usepackage{latexsym}
\usepackage{color}
\usepackage{dcolumn}
\usepackage{extarrows}

\newtheorem{theorem}{Theorem}
\newtheorem{lemma}[theorem]{Lemma}

\newtheorem{example}{Example}
\newtheorem{remark}{Remark}

\begin{document}
\begin{center}
\textbf{\Large{A kind of quaternary sequences of period $2p^mq^n$ and their linear complexity}}\footnote {the work was supported by
  National Science Foundation of
China No. 61602342 and No. 11701553,  Natural Science Foundation of Tianjin  under grant No. 18JCQNJC70300, the Science and
Technology Development Fund of Tianjin Education Commission
for Higher Education  2018KJ215, and the China Scholarship Council (No. 201809345010), Key Laboratory of Applied Mathematics of Fujian Province University
(Putian University) (No. SX201804 and SX201904), NFSC  No. 61872359, 61972456,  61802281, NSFT  No.16JCYBJC42300, the Science and
Technology Development Fund of Tianjin Education Commission
for Higher Education, No. 2017KJ213   and Foundation of Science
and Technology on Information Assurance Laboratory (No.
61421120102162112007) .\\
}
\end{center}

\begin{center}
\small Qiuyan Wang$^{a,b}$, Chenhuang Wu$^{b,c}$, Minghui Yang$^d$, Yang Yan$^{b,e,*}$
\end{center}

\begin{center}
\textit{\footnotesize  a.  School of Computer Science and
Technology, Tiangong University, Tianjin,
300387, China.\\
b. Provincial Key Laboratory of Applied Mathematics, Putian University, Putian, Fujian 351100, China\\
c. School of Computer Science and Engineering, University of Electronic Science and Technology of China, Chengdu, Sichuan 611731,  China\\
d. State Key Laboratory of Information Security, Institute of Information Engineering, Chinese Academy of Sciences, Beijing 100195,  China\\
e. School of Information Technology Engineering, Tianjin University of Technology and Education, Tianjin, 300222, China\\
yanyangucas@126.com
}
\end{center}


\noindent\textbf{Abstract:} Sequences with high linear complexity have wide applications in cryptography. In this paper, a new class of quaternary sequences over $\mathbb{F}_4$ with period $2p^mq^n$ is constructed using generalized cyclotomic classes. Results show that the linear complexity of these sequences attains the maximum.

\noindent\textbf{keywords:}  linear complexity, generalized cyclotomic classes, quaternary sequences

\section{Introduction}
Stream ciphers divide the plain-text into characters and encipher each character with a time-varying function. It is konwn that stream cipher   plays a dominant role in cryptographic practice and remains a crucial role in military and commercial secrecy systems. 
The security of   stream ciphers now depends on the ``randomness" of the key stream \cite{s7}.  For the system to be secure, the key stream must have a series of properties: balance, long period, low correlation and et al.

 A necessary requirement for unpredictability is a large linear complexity of the key stream, which is defined to be the length of the shortest linear feedback shift register able to produce the key stream. Let $\mathbb{F}_{l}$ denote a finite field with $l$ elements, where $l$ is a prime power. A sequence $S=\{s_i\}$ is periodic if there exists a positive integer $T$ such that $s_{j+T}= s_{j}$ for all $j\geq0$.
   Let $S=\{s_{i}\}$ be a periodic sequence over  $\mathbb{F}_{l}$.  The linear complexity of  $S$, denoted by $LC(S)$, is the least integer $L$ of a linear recurrence relation over $\mathbb{F}_{l}$  satisfied by $S$,
  $$
  -c_0s_{i+L}=c_{1}s_{i+L-1}+\ldots+c_Ls_{i}\ \ \  \textrm{\ for\ }\ i\geq0,
  $$
  where $c_0\neq0$, $c_0,c_1,\ldots,c_{L-1},c_L\in \mathbb{F}_{l}$. By  B-M algorithm \cite{DXS}, if $LC(S)\geq N/2$ ($N$ is the least period of $S$), then $S$ is considered to be good from the viewpoint of linear complexity.

 Periodic sequences have intensively studied  in the past few years, since they are widely used in CDMA (Code Division Multiple Access), global position systems and stream ciphers.  As special cases, cyclotomic and generalized cyclotomic sequences of different periods and orders have attracted many researchers to deeply explore due to their good pseudo-random  cryptographic properties \cite{AW,CD0,Hu}. In particular, the linear complexity of Legendre sequences and cyclotomic sequences of order $r$ were studied in \cite{DHS} and \cite{DH1}, respectively. Generalized cyclotomy, as a natural generalization of cyclotomy, was presented by Whiteman in \cite{ALW} and Ding et al. in \cite{CDH}. It should be noted that Whiteman's generalized cyclotomy is not in accordance with the classic cyclotomy.   Ding-Helleseth cyclotomy includes the classic cyclotomy as a special case. Whereafter, the linear complexity of generalized cyclotomic sequences have been determined \cite{V1,KJ,YD1,WJ1,WJ3,WJ2}.

Quaternary sequences are also important from the point of many practical applications; please refer to \cite{K}.  Owing to the nice algebraic structure, quaternary sequences also have received a lot of attentions. For instance, a kind of almost quaternary cyclotomic sequences was defined in \cite{TL} and was proved to have ideal autocorrelation property \cite{TL}. A new class of quaternary sequences of length $pq$, constructed by the  inverse Gray mapping, was studied in \cite{YK}.  A family of quaternary sequences of period $2p$ over $\mathbb{F}_4$ was presented and showed to possess high linear complexity \cite{Du}.


Motivated by the idea in \cite{Ke1,Ke2},  we constructed a new class of quaternary sequences over $\mathbb{F}_4$ with period $2p^mq^n$ by using the generalized cyclotomic classes in this paper.   From the definition of  $S$ in (\ref{eq-1.2}), we can easily  see that the new proposed sequences have longer period  contrast to these in \cite{Ke2}.  The linear complexity of these sequences is computed and results  show that the proposed sequences have  high linear complexity.

This paper is organized as follows. In Section $2$, the periodic sequence $S$ with period $2p^{m}q^{n}$ is given. Section $3$ determines the linear complexity of the constructed sequence. Finally, we give some remarks on this paper.

\section{Preliminaries}
For a positive integer $a\geq2$, use $\mathbb{Z}_{a}$ to denote the ring $\mathbb{Z}_{a}=\{0,1,2,\ldots,a-1\}$ with integer addition modulo $a$ and integer multiplication modulo $a$. Usually, we use $\mathbb{Z}_{a}^{*}$ to denote all invertible elements of $\mathbb{Z}_{a}$, i.e.,  all elements $b$ in $\mathbb{Z}_{a}$ satisfying $\gcd(a,b)=1$. Obviously, the group $\mathbb{Z}_{a}^{*}$ has cardinality $\phi(a)$, where $\phi(\cdot)$ denotes the Euler function.

For a subset $A \subset \mathbb{Z}_{a}$ and an element $b\in\mathbb{Z}_{a}$, define
\begin{align*}
b+A=A+b=\{a+b:a\in A\},\ \ \ \ \ bA=\{ab:a\in A\},
\end{align*}
where addition and multiplication refer to those in $\mathbb{Z}_{a}$.

Let $p$ and $q$ be two distinct odd primes.  Let  $m$ and $n$ denote  two positive integers. Suppose that $g_1$ is a primitive element of $\mathbb{Z}_{p^2}^{*}$. Then $g_1$ is  a primitive root of $\mathbb{Z}_{p^m}^{*}$ for $m\geq1$ \cite{B}. Without loss of generality,  assume  $g_1$ is an odd integer. It is known that $g_1$ is also a primitive root of $\mathbb{Z}_{2p^m}^{*}$ \cite{B}. Obviously, $g_1$ is a common primitive root of $\mathbb{Z}_{p^{i}}$ and $\mathbb{Z}_{2p^{i}}$ for all $1\leq i\leq m$. By the same argument, there exists an integer $g_2$ such that $g_2$ is a common primitive root of $\mathbb{Z}_{q^{j}}$ and $\mathbb{Z}_{2q^{j}}$ for any $1\leq j\leq n$.

\begin{lemma}[\cite{PSD}] \label{lem-1}
Let $m_1,\cdots,m_t$ be positive integers. For a set of integers $a_1,\cdots,a_t$, the system of congruences
$$
x\equiv a_i \bmod{m_i}, \ \ i=1,\cdots t
$$
has solutions if and only if
\begin{equation}\label{eq-1.1}
a_i\equiv a_{j}\bmod{\gcd(m_i,m_j)}, \ \ i\neq j, \ 1\leq i,j\leq t.
\end{equation}
If \eqref{eq-1.1} is satisfied, the solution is unique modulo $\textrm{lcm}(m_1,\cdots,m_t)$.
\end{lemma}
Let $g$ be the unique solution of the following congruence equations
\begin{align*}
\left\{
\begin{array}{l}
  g\equiv g_1\pmod{2p^{m}}, \\
  g\equiv g_2\pmod{2q^{n}}.
\end{array}
\right.
\end{align*}
Lemma \ref{lem-1} guaranteed the existence and uniqueness of the common primitive root $g$ of $p^{i}$, $2p^{i}$, $p^{j}$ and $2p^{j}$. Similarly, there exists a unique integer $y$ satisfying the following system of  congruences
\begin{align}\label{eq-y}
\left\{
\begin{array}{c}
   y\equiv g\pmod{2p^{m}},\\
  y\equiv1\pmod{2q^{n}}.
\end{array}
\right.
\end{align}

Assume that $e_{i,j}=\gcd(p^{i-1}(p-1), q^{j-1}(q-1))$ and $d_{i,j}=\frac{(p-1)p^{i-1}(q-1)q^{j-1}}{e_{i,j}}$. Then  $d_{i,j}$ is the least positive integer that satisfies $g^{d_{i,j}}\equiv 1(\bmod p^iq^j)$ (\cite{CDH}, Lemma 2), i.e., $\textrm{ord}_{p^iq^j}(g)=d_{i, j}$. In the sequel, let $i$ and $j$ be two integers with $1\leq i\leq m$ and $1\leq j\leq n$. The  generalized cyclotomic classes with respect to $p^iq^j$,  similar to the Ding- Helleseth's generalized cyclotomic classes (\cite{CDH}), are defined  as follows
\begin{align}
&D_0^{(p^iq^j)}=\left\{g^{2t}y^k\textrm{\ mod\ }p^iq^j: t=0, 1, \ldots, d_{i,j}/2-1; k=0, 1, \ldots, e_{i,j}-1\right\}, \label{eq-D0}\\
&D_1^{(p^iq^j)}=gD_0^{(p^iq^j)}(\bmod \ p^iq^j)\label{eq-D1}.
\end{align}
By Lemma 7 in \cite{ALW}, we get $\mathbb{Z}_{p^{i}q^{j}}^{\ast}=D_0^{(p^iq^j)}\bigcup D_1^{(p^iq^j)}$.
Let
\begin{align*}
&D_0^{(2p^iq^j)}=\left\{g^{2t}y^k \textrm{\ mod\ }2p^iq^j: t=0, 1, \ldots, d_{i,j}/2-1; k=0, 1, \ldots, e_{i,j}-1\right\},\\
&D_1^{(2p^iq^j)}=gD_0^{(2p^iq^j)}(\bmod \ 2p^iq^j).
\end{align*}
Similarly, we have $\mathbb{Z}_{2p^{i}q^{j}}^{\ast}=D_0^{(2p^iq^j)}\bigcup D_1^{(2p^iq^j)}$. For abbreviation, denote $H_h^{(2p^iq^j)}=p^{m-i}q^{n-j}D_h^{(2p^iq^j)}$ and $H_h^{(p^iq^j)}=p^{m-i}q^{n-j}D_h^{(p^iq^j)}$ for $h=0,1$.
With the above preparations, we get a partition of $\mathbb{Z}_{2p^{m}q^{n}}$ as follows:
\begin{equation*}\begin{split}
\mathbb{Z}_{2p^{m}q^{n}}  = & \bigcup_{i=1}^{m}\bigcup_{j=1}^{n}\bigcup_{h=0}^{1}p^{m-i}q^{n-j}\left(D_h^{(2p^iq^j)}\bigcup2D_h^{(p^iq^j)}\right)\\
                                    & \bigcup_{j=1}^{n}\bigcup_{h=0}^{1}p^{m}q^{n-j}\left(D_h^{(2q^j)}\bigcup2 D_h^{(q^j)}\right)\\
                                    & \bigcup_{i=1}^{m}\bigcup_{h=0}^{1}p^{m-i}q^{n}\left(D_h^{(2p^i)}\bigcup2 D_h^{(p^i)}\right)\\
                                    &\bigcup\left\{0, p^{m}q^{n}\right\}.
\end{split}\end{equation*}

Let $\mathbb{F}_4=\{0, 1, \alpha, \alpha^2\}$ be the finite field with $4$ element, where $\alpha$ satisfies $\alpha^{2}=\alpha+1$. A class of quaternary sequence can be given by allocating each elements of $\mathbb{F}_4$ to each generalized cyclotomic class with respect to $2p^{m}q^{n}$. To ensure the constructed sequence has high linear complexity, we should technologically do with it.

Let $\{ a, b, c, d\}$ be a set of four tuple over $\mathbb{F}_4$ and the elements in this tuple are pairwise distinct. A quaternary generalized cyclotomic sequence $S=\{s_i\}$ of period
$2p^{m}q^{n}$ is  defined as
\begin{align}\label{eq-1.2}
 s_i
 = \left\{ \begin{array}{ll}
0, & \textrm{if $i=0 $},\\
e, & \textrm{if $i=p^mq^n$},\\
a, & \textrm{if $i\in \bigcup_{i=1}^m\bigcup_{j=1}^nH_0^{(2p^iq^j)}\bigcup_{i=1}^mH_0^{(2p^i)}\bigcup_{j=1}^nH_0^{(2q^j)}$},\\
b, & \textrm{if $i\in \bigcup_{i=1}^m\bigcup_{j=1}^nH_1^{(2p^iq^j)}\bigcup_{i=1}^mH_1^{(2p^i)}\bigcup_{j=1}^nH_1^{(2q^j)}$},\\
c, & \textrm{if $i\in \bigcup_{i=1}^m\bigcup_{j=1}^n2H_0^{(p^iq^j)}\bigcup_{i=1}^m2H_0^{(p^i)}\bigcup_{j=1}^n2H_0^{(q^j)}$},\\
d, & \textrm{if $i\in \bigcup_{i=1}^m\bigcup_{j=1}^n2H_1^{(p^iq^j)}\bigcup_{i=1}^m2H_1^{(p^i)}\bigcup_{j=1}^n2H_1^{(q^j)}$},
\end{array} \right.
\end{align}
where $e\neq b+d$ and $e\in\mathbb{F}_4^{\ast}$ if $p\equiv \pm 1\pmod 8$, $e\not\in \{b,b+c\}$ and $e\in\mathbb{F}_4^{\ast}$ if $p\equiv \pm 3\pmod 8$. It is easily seen that the sequence $S=\{s_i\}$ is balanced.

\section{Linear complexity of the constructed sequences}
In generating of running keys, linear feedback shift register (LFSR) is one of the most useful devices. And it is shown that every periodic sequence can be generated by LFSR. For researchers, what they most concern is the shortest length of  LFSR that could produce a given sequence $S$, which is referred  to the linear complexity of $S$.

Let $S=\{s_i\}$ be a periodic sequence over the finite field $\mathbb{F}_l$ of period $N$. We first recall  the definition of linear complexity of periodic sequences that is given in Sect. $1$.  The \emph{linear complexity} of $S$ over $\mathbb{F}_{l}$, denoted by  $LC(S)$, is the smallest positive integer $L$ satisfying the following linear recurrence relation
\begin{equation}\label{eq-1.3}
-c_0s_{i+L}=c_{1}s_{i+L-1}+\cdots+c_Ls_{i},\ \  \textrm{\ for\ }\ i\geq0,
\end{equation}
where $c_0\neq0$, $c_0,c_1,\ldots,c_{L-1},c_L\in \mathbb{F}_{l}$. The polynomial
$$
c(x)=c_Lx^{L}+c_{L-1}x^{L-1}+\cdots+c_{1}x+c_{0}\  \in \mathbb{F}_{l}[x]
$$
associated with the linear recurrence relation \eqref{eq-1.3} is called the characteristic polynomial of $S$. A characteristic
polynomial with the smallest degree is called a \emph{minimal polynomial} of $S$ \cite{DXS}.
For the periodic sequence $S$, let $S(x)=s_0+s_1x+\cdots + s_{N-1}x^{N-1}\in \mathbb{F}_{l}[x]$, which is called the \emph{generating polynomial} of $S$.
The following lemma gives a method to computer the linear complexity of $S$ by using the generating polynomial $S(x)$.
\begin{lemma}[\cite{CDR}]\label{lem1}
Let $S$ be a  sequence over $\mathbb{F}_l$ of period $N$.  Then the
minimal polynomial $m(x)$ of $S$  is
$$
m(x)=\frac{x^{N}-1}{\gcd(x^{N}-1,S(x))},
$$
and the linear complexity $LC(S)$ of $S$ is given by
\begin{equation*}
N-\deg\left(\gcd(x^{N}-1,S(x))\right),
\end{equation*}
where $S(x)$ is the generating polynomial of $S$.
\end{lemma}
\begin{lemma}[Lemma 2, \cite{V1} and Lemma 1, \cite{Ke1}]\label{lem2}
Let notations be defined as above. Then for $1\leq i\leq m$ and $h=0,1$, we have
\begin{enumerate}
\item  $D_{h}^{(p^{i})}=\{x+py:x\in D_h^{(p)},y\in \mathbb{Z}_{p^{i-1}}\}$;
\item $D_{h}^{(2p^{i})}=\{x+py+\delta_{x,y}:x\in D_h^{(p)},y\in \mathbb{Z}_{p^{i-1}}\}$,
where
$$
\delta_{x,y}=\left\{\begin{array}{ll}
                      0, & \textrm{if\ $x+py$\ is odd},  \\
                      p^{i}, & \textrm{otherwise}.
                    \end{array}
                    \right.
$$
\end{enumerate}
\end{lemma}


For the generalized cyclotomic classes $D_h^{(p^iq^j)}$ and $D_h^{(2p^iq^j)}$ corresponding to $p^iq^j$ and $2p^{i}q^{j}$, we have the following lemma.
\begin{lemma}[Lemma1, \cite{Yang}]\label{lem-2}
 For $1\leq i\leq m$ and $1\leq j\leq n$, we have
\begin{enumerate}
\item  $D_h^{(p^iq^j)}=\left\{a+pqb: a\in D_h^{(pq)}, b\in \mathbb{Z}_{p^{i-1}q^{j-1}}\right\}$;
\item  $D_h^{(2p^iq^j)}=\left\{a+pqb+\delta_{a,b}: a\in D_h^{(pq)}, b\in \mathbb{Z}_{p^{i-1}q^{j-1}}\right\},$
where
\begin{equation*}\begin{split}
\delta_{a,b}=\begin{cases} 0, & \textrm{if $a+bpq$ is odd,} \\
p^iq^j, & \textrm{otherwise.}
\end{cases}
\end{split}\end{equation*}
\end{enumerate}
\end{lemma}
\begin{lemma}[\cite{CDH}]\label{lem-3}
$2\in D_{h}^{(p)}$ if and only if $2\in D_{h}^{(p^{i})}$ for $1\leq i\leq m$ and $h=0,1$.
\end{lemma}
\begin{lemma}[\cite{B}]\label{lem-4}
Let symbols be the same as before. Then we have
\begin{enumerate}
\item $2\in D_{0}^{(p)}$ if and only if $p\equiv \pm1\pmod{8}$;
\item $2\in D_{1}^{(p)}$ if and only if  $p\equiv\pm3\pmod{8}$.
\end{enumerate}
\end{lemma}
\begin{lemma}\label{lem-5}
Let symbols be the same as before. Then we have
\begin{enumerate}
\item $2\in D_0^{(pq)}$ if and only if $q\equiv \pm 1\pmod{8}$;
\item $2\in D_1^{(pq)}$ if and only if $q\equiv \pm 3\pmod{8}$.
\end{enumerate}
\end{lemma}
\begin{proof}
We only prove the first part of this lemma.\\
Sufficiency: Since $\gcd(2,pq)=1$, then $2\in \mathbb{Z}_{pq}^{*}$. If $q\equiv \pm 1\pmod{8}$, by Lemma \ref{lem-4}, $2\in D_{0}^{(q)}$. Since $D_{0}^{(q)}\subseteq D_{0}^{(pq)}$, we get $2\in D_{0}^{(pq)}$.

Necessity: If $2\in D_0^{(pq)}$, by the definitions of $y$ in \eqref{eq-y} and $D_0^{(pq)}$ in \eqref{eq-D0}, we know $2\pmod{q}\in D_{0}^{(q)}$. It follows from Lemma \ref{lem-4} that $q\equiv \pm 1\pmod{8}$.

By the method analogous to that used above, we can get the second conclusion of this lemma. 
\end{proof}
Let $d=\ $ord$_{p^mq^n}(4).$  Assume that $\beta$ is a primitive $p^mq^n$th root of unity in $\mathbb{F}_{4^d}$. It can be easily checked that
\begin{equation}\label{eq}
x^{p^mq^n}-1=(x-1)(x-\beta)\ldots(x-\beta^{p^mq^n-1}).
\end{equation}
By Lemma \ref{lem1}, in order to determine the linear complexity of $S$, we need to determine $\gcd(x^{2p^mq^n}-1, S(x))=\gcd((x^{p^mq^n}-1)^2, S(x))$ over $\mathbb{F}_4[x]$. By \eqref{eq}, we should check whether $\beta^{i}$, $0\leq i\leq p^{m}q^{n}-1$, is a root of $S(x)$. If it is a root of $S(x)$, we need to verify whether it is a multiple root of $S(x)$.

Recall that $H_h^{(2p^iq^j)}=p^{m-i}q^{n-j}D_h^{(2p^iq^j)}$ and $H_h^{(p^iq^j)}=p^{m-i}q^{n-j}D_h^{(p^iq^j)}$ for $h=0,1$.  Define
\begin{align*}
&S_h^{(i,j)}(x)=\sum_{t\in H_{h}^{(2p^iq^j)}}x^t,\\
&S_{h}^{(i,0)}(x)=\sum_{t\in H_h^{(2p^i)}}x^t,\\
&S_h^{(0,j)}(x)=\sum_{t\in H_{h}^{(2q^j)}}x^t,
\end{align*}
for $1\leq i\leq m$, $1\leq j\leq n$ and $0\leq h\leq1$. Let $a$ and $b$ be two integers with $0\leq a\leq m-1$ and $0\leq b\leq n-1$. For $1\leq k\leq p^mq^n-1$, suppose $k=p^aq^bl$ with $\gcd(l, pq)=1.$  It follows from Lemma \ref{lem-2} that
\begin{align}
S_h^{(i,j)}(\beta^{k})&=\sum_{t\in p^{m-i}q^{n-j}D_{h}^{(2p^{i}q^{j})}}\beta^{kt}\nonumber \\
&=\sum_{t\in p^{m+a-i}q^{n+b-j}lD_{h}^{(2p^{i}q^{j})}}\beta^{t} \nonumber  \\
&=\sum_{t\in p^{m+a-i}q^{n+b-j}lD_{h}^{(p^{i}q^{j})}}\beta^{t} \nonumber \\
&=\sum_{t\in H_{h}^{(p^{i}q^{j})}}\beta^{kt}. \label{eq-3.1}
\end{align}
Similarly, we have
\begin{align}
&S_h^{(i,0)}(\beta^{k})=\sum_{t\in H_{h}^{(p^{i})}}\beta^{kt},  \label{eq-3.2}\\
&S_h^{(0,j)}(\beta^{k})=\sum_{t\in H_{h}^{(q^{j})}}\beta^{kt}.  \label{eq-3.3}
\end{align}
Combining \eqref{eq-3.1}, \eqref{eq-3.2} and \eqref{eq-3.3}, we have
 \begin{equation}\label{eq-3.03}
 \begin{split}
S(\beta^k)=&\sum_{i=0}^{p^mq^n-1}s_{i}x^{i}\\
 =& e\beta^{p^mq^n}+a\bigg(\sum_{i=1}^m\sum_{j=1}^nS_0^{(i,j)}(\beta^k)+\sum_{i=1}^mS_0^{(i,0)}(\beta^k)+ \sum_{j=1}^nS_0^{(0,j)}(\beta^k)\bigg)\\
                                    & +b\bigg(\sum_{i=1}^m\sum_{j=1}^nS_1^{(i,j)}(\beta^{k})+\sum_{i=1}^mS_1^{(i,0)}(\beta^{k})+ \sum_{j=1}^nS_1^{(0,j)}(\beta^{k})\bigg)\\
                                    & +c\bigg(\sum_{i=1}^m\sum_{j=1}^nS_0^{(i,j)}(\beta^k)+\sum_{i=1}^mS_0^{(i,0)}\beta^k)+                                     \sum_{j=1}^nS_{0}^{(0,j)}(\beta^k)\bigg)^2\\
                                    &+d\bigg(\sum_{i=1}^m\sum_{j=1}^nS_1^{(i,j)}(\beta^{k})+\sum_{i=1}^mS_1^{(i,0)}(\beta^{k})+ \sum_{j=1}^nS_{1}^{(0,j)}(\beta^{k})\bigg)^2.
\end{split}\end{equation}

Let
$$A(\beta^k)=\sum_{i=1}^m\sum_{j=1}^nS_0^{(i,j)}(\beta^k)+\sum_{i=1}^mS_0^{(i,0)}(\beta^k)+\sum_{j=1}^nS_0^{(0,j)}(\beta^k).$$

Then \begin{equation*}\begin{split}
A(\beta^k)  &= \sum_{i=1}^m\sum_{j=1}^n\sum_{t\in p^{m-i}q^{n-j}D_0^{(p^iq^j)}}\beta^{kt}+\sum_{i=1}^m\sum_{t\in q^{n}p^{m-i}D_0^{(p^i)}}\beta^{kt}+\sum_{j=1}^n\sum_{t\in p^mq^{n-j}D_0^{(q^j)}}\beta^{kt}\\
                                    & =\sum_{i=1}^m\sum_{j=1}^n\sum_{t\in p^{m+a-i}q^{n+b-j}lD_0^{(p^iq^j)}}\beta^t+\sum_{i=1}^m\sum_{t\in p^{m-i+a}q^{n+b}lD_0^{(p^i)}}\beta^t\\
                                    &\ \ \ +\sum_{j=1}^n\sum_{t\in p^{m+a}q^{n+b-j}lD_0^{(q^j)}}\beta^t.
\end{split}\end{equation*}
We first compute
\begin{align*}
&S_0^{(i,j)}(\beta^k)=\sum_{t\in p^{m+a-i}q^{n+b-j}lD_0^{(p^iq^j)}}\beta^t, \ \
\end{align*}
where $k=p^aq^bl$  and $\gcd(pq,l)=1$. The computation is divided  into the following  cases.\\
Case 1): $i\leq a$ and $j\leq b$.  With simple derivation, we have
$$S_{0}^{(i,j)}(\beta^k)=\left|D_0^{(p^iq^j)}\right|=\frac{(p-1)(q-1)p^{i-1}q^{j-1}}{2}.$$
Case 2): $i=a+1, j=b+1$. Then
$$S_{0}^{(i,j)}(\beta^k)=\sum_{t\in D_0^{(pq)}}\zeta_{pq}^{lt},$$ where $\zeta_{pq}=\beta^{p^{m-1}q^{n-1}}$ is a $pq$th primitive root of unity.\\
Case 3): $i> a+1$ or $j> b+1$. Let $\eta=\beta^{p^{m+a-i}q^{n+b-j}}$, then $\eta^{pq}\neq 1$. It follows  from Lemma \ref{lem-2} that
$$S_{0}^{(i,j)}(\beta^k)=\sum_{t\in D_0^{(p^iq^j)}}\eta^t=\sum_{t_1\in D_0^{(pq)}}\eta^{t_1}\sum_{t_2\in \mathbb{Z}_{p^{i-1}q^{j-1}}}\eta^{pqt_2}=0.$$
Case 4): $i\leq a$, $j=b+1$. Let $\zeta_q=\beta^{p^{m+a-i}q^{n-1}}$, then  $\zeta_q$ is a $q$th primitive root of unity. Hence, we obtain
\begin{equation*}\begin{split}
S_0^{(i,j)}(\beta^k)  &= \sum_{t\in p^{m+a-i}q^{n+b-j}lD_0^{(p^iq^j)}}\beta^t\\
&=\sum_{t\in p^{m+a-i}q^{n-1}lD_0^{(p^iq^{b+1})}}\beta^t\\
& =\sum_{t\in lD_0^{(p^iq^{b+1})}}\zeta_q^t\\
&=(p-1)p^{i-1}q^b\sum_{t\in lD_0^{(q)}}\zeta_q^t\\
&=0.
                                    \end{split}\end{equation*}
Case 5): $i=a+1, j\leq b$. By Lemma \ref{lem-2}, we get
\begin{equation*}\begin{split}
S_0^{(i,j)}(\beta^k)  &= \sum_{t\in p^{m+a-i}q^{n+b-j}lD_0^{(p^iq^j)}}\beta^t\\
&=\sum_{t\in p^{m-1}q^{n+b-j}lD_0^{(p^iq^j)}}\beta^t\\
                                    & =\sum_{t\in lD_0^{(p^{a+1}q^b)}}\zeta_p^t\\
                                    &=\frac{(q-1)p^aq^{j-1}}{2}\sum_{t\in\mathbb{Z}_p^{\ast}}\zeta_p^t\\
                                    & =\frac{(q-1)p^aq^{j-1}}{2}.
                                    \end{split}\end{equation*}
where $\zeta_p=\beta^{p^{m-1}q^{n+b-j}}$.

From the above discussions, we have proved the first part of the following  lemma.

\begin{lemma}\label{lem-6}
 For $k=p^aq^bl$ with $\gcd(l, pq)=1, 0\leq a\leq m-1, 0\leq b\leq n-1$, we have
\begin{enumerate}
\item \begin{align}\label{eq-3.4}
\begin{split}
S_0^{(i,j)}(\beta^k)=\begin{cases} \frac{(p-1)(q-1)p^{i-1}q^{j-1}}{2}, & \textrm{if $i\leq a$\ and\ $j\leq b$, } \\ \sum_{t\in D_0^{(pq)}}\zeta_{pq}^{lt}, & \textrm{if $i=a+1$ and $j=b+1$},\\ 0, & \textrm{if $i>a+1$ or $j>b+1$},\\ 0, & \textrm{if $i\leq a$\ and\  $j=b+1$},\\ \frac{q-1}{2}, & \textrm{if $i=a+1$\ and\  $j\leq b$};
\end{cases}
                                    \end{split}\end{align}
\item \begin{align}\label{eq-3.5}
\begin{split}
S_1^{(i,j)}(\beta^k)=\begin{cases} \frac{(p-1)(q-1)p^{i-1}q^{j-1}}{2}, & \textrm{if $i\leq a$\ and \ $j\leq b$, } \\ \sum_{t\in D_1^{(pq)}}\zeta_{pq}^{lt}, & \textrm{if $i=a+1$\ and\  $j=b+1$},\\ 0, & \textrm{if $i> a+1$ or $j>b+1$},\\ 0, & \textrm{if $i\leq a$ and $j=b+1$},\\ \frac{q-1}{2}, & \textrm{if $i=a+1$ and $j\leq b$},
\end{cases}
                                    \end{split}\end{align}
                                    \end{enumerate}
\end{lemma}

\noindent where $\zeta_{pq}=\beta^{p^{m-1}q^{n-1}}$ is a $pq$th primitive root of unity and $\beta$ is a $p^mq^n$th primitive root of unity.
\begin{proof}
The proof of the second conclusion of this lemma is similar to the first part and we omit it.
\end{proof}
\begin{lemma} \label{lem-7}
For $k=p^aq^bl$ with $\gcd(pq, l)=1$, we obtain
\begin{enumerate}
\item
\begin{align}\label{eq-3.6}
\begin{split}
S_0^{(i,0)}(\beta^k)
=\begin{cases} \frac{p^{i-1}(p-1)}{2} & \textrm{if $i\leq a$ }, \\ \sum_{t\in D_0^{(p)}}\zeta_p^{lt}, & \textrm{if $i=a+1,$}\\ 0, & \textrm{if $i>a+1$ };
\end{cases}
                                    \end{split}\end{align}
\item
\begin{align}\label{eq-3.7}
\begin{split}
S_1^{(i,0)}(\beta^k)
=\begin{cases} \frac{(p-1)p^{i-1}}{2} & \textrm{if $i\leq a$, } \\ \sum_{t\in D_1^{(p)}}\zeta_p^{lt}, & \textrm{if $i=a+1$},\\ 0, & \textrm{if $i> a+1$ };
\end{cases}
                                    \end{split}\end{align}
\item
\begin{align}\label{eq-3.8}
\begin{split}
S_0^{(0,j)}(\beta^k)
=\begin{cases} \frac{(q-1)q^{j-1}}{2} & \textrm{if $j\leq b$, } \\ \sum_{t\in D_0^{(q)}}\zeta_q^{lt}, & \textrm{if $j=b+1$},\\ 0, & \textrm{if $j> b+1$ };
\end{cases}
                                    \end{split}\end{align}
\item
\begin{align}\label{eq-3.9}
\begin{split}
S_1^{(0,j)}(\beta^k)
=\begin{cases} \frac{(q-1)q^{j-1}}{2} & \textrm{if $j\leq b$, } \\ \sum_{t\in D_1^{(q)}}\zeta_q^{lt}, & \textrm{if $j=b+1$}\\ 0, & \textrm{if $j> b+1$ },
\end{cases}
                                    \end{split}\end{align}
                                    \end{enumerate}
where $\zeta_p=\beta^{p^{m-1}q^{n+b}}$ and $\zeta_q=\beta^{p^{m+a}q^{n-1}}$.
\end{lemma}
\begin{proof}
Because \eqref{eq-3.6}-\eqref{eq-3.9} can be proved in a similar way, here we only prove \eqref{eq-3.6}.
By notations and \eqref{eq-3.2}, we get
$$S_0^{(i,0)}(\beta^k)=\sum_{t\in p^{m-i}q^nD_0^{(2p^i)}}\beta^{kt}=\sum_{t\in p^{m+a-i}q^{n+b}lD_0^{(p^i)}}\beta^t,$$
where $k=p^aq^bl$ with $\gcd(pq,l)=1$. \\
If  $i\leq a$, for each $t\in p^{m+a-i}q^{n+b}lD_0^{(p^i)}$, it can be easily seen that $\beta^t=1$. Thus,  $$S_0^{(i,0)}(\beta^k)=|D_0^{(p^i)}|=\frac{p^{i-1}(p-1)}{2}.$$
If $i=a+1$, then $\beta^{p^{m-1}q^{n+b}}$ is a $p$th primitive root of unity and
$$S_0^{(i,0)}(\beta^k)=\sum_{t\in D_0^{(p)}}\zeta_p^{lt},$$
where $\zeta_p=\beta^{p^{m-1}q^{n+b}}$. \\
If $i>a+1$, we have $(\beta^{p^{m+a-i}q^{n+b}})^p\neq 1$ and $(\beta^{p^{m+a-i}q^{n+b}})^{p^i}=1$.
By Lemma \ref{lem2}, we know
$$D_0^{(p^i)}=\left\{x+py: x\in D_0^{(p)}, y\in \mathbb{Z}_{p^{i-1}}\right\}.$$
Therefore,
 $$S_0^{(i,0)}(\beta^k)=\sum_{t\in D_0^{(p^i)}}\eta^t=\sum_{t_1\in D_0^{(p)}}\eta^{t_1}\sum_{t_2\in \mathbb{Z}_{p^{i-1}}}\eta^{t_2p}=\sum_{t_1\in D_0^{(p)}}\eta^{t_1}\cdot 0=0,$$
where $\eta=\beta^{p^{m+a-i}q^{n+b}l}$ and $\eta^p\neq 1.$
\end{proof}

In the following, we will determine the terms with $a, b, c, d$ as coefficients in \eqref{eq-3.03}, respectively.

First, we compute the terms with $a$ as coefficient.

It follows from Lemmas \ref{lem-6} and \ref{lem-7} that

\begin{equation*}\begin{split}
& a\bigg(\sum_{i=1}^m\sum_{j=1}^nS_0^{(i,j)}(\beta^k)+\sum_{i=1}^mS_0^{(i,0)}(\beta^k)+\sum_{j=1}^nS_0^{(0,j)}(\beta^k)\bigg)\\
                                    =& a\bigg(\sum_{i=1}^m\sum_{j=1}^bS_0^{(i,j)}(\beta^k)+\sum_{i=1}^m\sum_{j=b+1}S_0^{(i,b+1)}(\beta^k)+\sum_{i=1}^m\sum_{j>b+1}S_0^{(i,j)(\beta^k)}\\
                                    &\ +\sum_{i=1}^aS_0^{(i,0)}(\beta^k)+\sum_{i=a+1}S_0^{(a+1,0)}(\beta^k)+0+\frac{q^{j-1}(q-1)}{2}+\sum_{t\in D_0^{(q)}}\zeta_q^{lt}\bigg)\\
                                    =&a\bigg(\frac{q-1}{2}+\frac{(q-1)q^{j-1}}{2}+\frac{p^{i-1}(p-1)}{2}+\sum_{t\in D_0^{(pq)}}\zeta_{pq}^{lt}+\sum_{t\in D_0^{(p)}}\zeta_{p}^{lt}+\sum_{t\in D_0^{(q)}}\zeta_{q}^{lt}\bigg)\\
                                    =&a\bigg(\frac{p-1}{2}+\sum_{t\in D_0^{(pq)}}\zeta_{pq}^{lt}+\sum_{t\in D_0^{(p)}}\zeta_{p}^{lt}+\sum_{t\in D_0^{(q)}}\zeta_{q}^{lt}\bigg).
                                   \end{split}\end{equation*}

Similarly, we compute the terms with $b$ as coefficient :

\begin{equation*}\begin{split}
& b\bigg(\sum_{i=1}^m\sum_{j=1}^bS_1^{(i,j)}(\beta^k)+\sum_{i=1}^m\sum_{j=b+1}S_1^{(i,b+1)}(\beta^k)+0+\frac{p^{i-1}(p-1)}{2}\\
                                    &\ +\sum_{t\in D_1^{(p)}}\zeta_p^{lt}+0+\frac{q^{j-1}(q-1)}{2}+\sum_{t\in D_1^{(q)}}\zeta_q^{lt}+0\bigg)\\
                                  =&b\bigg(\frac{p-1}{2}+\sum_{t\in D_1^{(pq)}}\zeta_{pq}^{lt}+
\sum_{t\in D_1^{(p)}}\zeta_{p}^{lt}+\sum_{t\in D_1^{(q)}}\zeta_{q}^{lt}\bigg).
                                   \end{split}\end{equation*}

The terms with $c$ as coefficient are

$$c\bigg(\left(\frac{p-1}{2}\right)^2+\sum_{t\in D_0^{(pq)}}\zeta_{pq}^{2lt}+\sum_{t\in D_0^{(p)}}\zeta_{p}^{2lt}+\sum_{t\in D_0^{(q)}}\zeta_{q}^{2lt}\bigg).$$

The terms with $d$ as coefficient are
$$d\bigg(\left(\frac{p-1}{2}\right)^2+\sum_{t\in D_1^{(pq)}}\zeta_{pq}^{2lt}+\sum_{t\in D_1^{(p)}}\zeta_{p}^{2lt}+\sum_{t\in D_1^{(q)}}\zeta_{q}^{2lt}\bigg).$$

It can be easily checked that
\begin{align*}
S(1)&=e+(a+b+c+d)\frac{(p-1)(q-1)p^{i-1}q^{j-1}+(p-1)p^{i-1}+(q-1)q^{j-1}}{2}\\
&=e\neq0.
\end{align*}
Next, we determine $S(\beta^k)$  according to the values of $p$ and $q$, where $k=p^aq^bl$ with $\gcd(pq,l)=1$.

(1) If $p\equiv \pm 1\pmod 8$ and $q\equiv \pm 1\pmod 8$, by Lemmas \ref{lem-4} and \ref{lem-5} we know $2\in D_0^{(pq)}$, $2\in D_0^{(p)}$ and  $2\in D_0^{(q)}$. Hence,

\begin{equation*}\begin{split}
 S(\beta^k)=&e+\frac{p-1}{2}(a+b+c+d)+(a+c)\sum_{t\in D_0^{(pq)}}\zeta_{pq}^{lt}+(b+d)\sum_{t\in D_1^{(pq)}}\zeta_{pq}^{lt}\\
                                    &+(a+c)\sum_{t\in D_0^{(p)}}\zeta_{p}^{lt}+(b+d)\sum_{t\in D_1^{(p)}}\zeta_{p}^{lt}+(a+c)\sum_{t\in D_0^{(q)}}\zeta_{q}^{lt}+(b+d)\sum_{t\in D_1^{(q)}}\zeta_{q}^{lt}\\
                                    =& e+b+d.
                                   \end{split}\end{equation*}

(2) If $p\equiv \pm 3\pmod 8$ and $q\equiv \pm 1\pmod 8$, then $2\in D_0^{(pq)}$, $2\in D_0^{(q)}$ and $2\in D_1^{(p)}$. Hence,

\begin{equation*}\begin{split}
 S(\beta^k)=&e+(a+c)\sum_{t\in D_0^{(pq)}}\zeta_{pq}^{lt}+(b+d)\sum_{t\in D_1^{(pq)}}\zeta_{pq}^{lt}+(a+c)\sum_{t\in D_0^{(q)}}\zeta_{q}^{lt}\\
                                    &+(b+d)\sum_{t\in D_1^{(q)}}\zeta_{q}^{lt}+(a+c+d)\sum_{t\in D_1^{(p)}}\zeta_{p}^{lt}+b\left(1+\sum_{t\in D_1^{(p)}}\zeta_{p}^{lt}\right)\\
                                    =& e+b.
                                   \end{split}\end{equation*}

 (3) If $p\equiv \pm 1\pmod 8$ and $q\equiv \pm 3\pmod 8$, then $2\in D_1^{(pq)}$, $2\in D_1^{(q)}$ and $2\in D_0^{(p)}$. Hence,

\begin{equation*}\begin{split}
 S(\beta^k)=&e+(a+d)\sum_{t\in D_0^{(pq)}}\zeta_{pq}^{lt}+(b+c)\sum_{t\in D_1^{(pq)}}\zeta_{pq}^{lt}+(a+c)\sum_{t\in D_0^{(p)}}\zeta_{p}^{lt}\\
                                    +&(b+d)\sum_{t\in D_1^{(p)}}\zeta_{p}^{lt}+(a+d)\sum_{t\in D_0^{(q)}}\zeta_{q}^{lt}+(b+c)\sum_{t\in D_1^{(q)}}\zeta_{q}^{lt}\\
                                    =& e+b+d.
                                   \end{split}\end{equation*}

  (4) If $p\equiv \pm 3\pmod 8$ and $q\equiv \pm 3\pmod 8$, then $2\in D_1^{(pq)}$, $2\in D_1^{(q)}$ and $2\in D_1^{(p)}$. Hence,

\begin{equation*}\begin{split}
 S(\beta^k)=&e+(a+d)\sum_{t\in D_0^{(pq)}}\zeta_{pq}^{lt}+(b+c)\sum_{t\in D_1^{(pq)}}\zeta_{pq}^{lt}+(a+d)\sum_{t\in D_0^{(p)}}\zeta_{p}^{lt}\\
                                    &+(b+c)\sum_{t\in D_1^{(p)}}\zeta_{p}^{lt}+(a+d)\sum_{t\in D_0^{(q)}}\zeta_{q}^{lt}+(b+c)\sum_{t\in D_1^{(q)}}\zeta_{q}^{lt}\\
                                    =& e+b+c.
                                   \end{split}\end{equation*}

 From the choice of $e$, we know  $e\neq b+d$ if $p\equiv\pm 1\pmod 8$, and  $e\notin\{b, b+c\}$  if $p\equiv\pm 3\pmod 8$.
By Lemma \ref{lem1} and the above discussions, we obtain  $LC(S)=2p^mq^n$.
\begin{theorem}
Let $S=\{s_i\}$ be the quaternary sequence defined by (\ref{eq-1.2}). Then the linear complexity of  $S$ is  $2p^mq^n$.
\end{theorem}
\begin{example}
 Let $(p,q,m,n)=(3,5,1,1)$ and  $(a,b,c,d,e)=(\alpha,1+\alpha,1,0,1)$. Then
\begin{align*}
& D_0^{(2pq)}\bigcup qD_0^{(2p)}\bigcup D_0^{(2q)}=\{1,3,5,11,19,27,29\}. \\
& D_1^{(2pq)}\bigcup qD_1^{(2p)}\bigcup pD_1^{(2q)}=\{7,9,13,17,21,23,25\}.\\
& 2D_0^{(pq)}\bigcup 2qD_0^{(p)}\bigcup 2pD_0^{(q)}=\{2,6,8,10,22,24,28\}.\\
& 2D_1^{(2pq)}\bigcup 2qD_0^{(p)}\bigcup 2pD_0^{(q)}=\{4,12,14,16,18,20,26\}.
\end{align*}
 It can be checked by Magma that $\gcd(x^{15}-1,S(x))=1$ and $LC(S)=30$.
\end{example}
\begin{example}
 Let $(p,q,m,n)=(3,7,1,1)$ and  $(a,b,c,d,e)=(\alpha,1+\alpha,1,0,1)$. Then
 \begin{align*}
& D_0^{(2pq)}\bigcup qD_0^{(2p)}\bigcup D_0^{(2q)}=\{1,3,7,11,23,25,27,29,33,37\}. \\
& D_1^{(2pq)}\bigcup qD_1^{(2p)}\bigcup pD_1^{(2q)}=\{5,9,13,15,17,19,31,35,39,41\}.\\
& 2D_0^{(pq)}\bigcup 2qD_0^{(p)}\bigcup 2pD_0^{(q)}=\{2,4,6,8,12,14,16,22,24,32\}.\\
& 2D_1^{(2pq)}\bigcup 2qD_0^{(p)}\bigcup 2pD_0^{(q)}=\{10,18,20,26,28,30,34,36,38,40\}.
\end{align*}
 It can be checked by Magma that $\gcd(x^{21}-1,S(x))=1$ and $LC(S)=42$.
\end{example}

\begin{remark}
For $1\leq k\leq p^mq^n-1$, let $k=p^aq^bl$ with $\gcd(l, ab)=1.$
In the case that $e=b+d\in \mathbb{F}_4^{\ast}$ if $p\equiv \pm 1\pmod 8$, and the case that $e\in \{b, b+c\}$ if $p\equiv\pm 3\pmod 8$, we know $S(\beta^{k})=0$ for $1\leq k\leq p^mq^n-1$. Hence, we need to check if $\beta^{k}$ is a multiple root of $S(x)$. This means we should check if $\beta^k$ is a root of the derivation polynomial $S^{'}(x)$ of the generating polynomial $S(x)$ of $S$. By definitions,
 we have
$$\beta^k S'(\beta^k)=e+b+(a+b)\sum_{t\in D_0^{(pq)}}\zeta_{pq}^{lt}+(a+b)\sum_{t\in D_0^{(p)}}\zeta_{p}^{lt}+(a+b)\sum_{t\in D_0^{(q)}}\zeta_{q}^{lt},$$
and
 $$\beta^k S'(\beta^k)=e+(a+b)\sum_{t\in D_0^{(pq)}}\zeta_{pq}^{lt}+(a+b)\sum_{t\in D_0^{(p)}}\zeta_{p}^{lt}+(a+b)\sum_{t\in D_0^{(q)}}\zeta_{q}^{lt},$$
respectively.

For the fixed $a$ and $b$ with $ 0\leq a\leq m-1$ and $0\leq b\leq n-1,$ by the following equations
\begin{align*}
&\sum_{t\in D_0^{(pq)}}\zeta_{pq}^t+\sum_{t\in D_1^{(pq)}}\zeta_{pq}^t=1,\\
&\sum_{t\in D_0^{(p)}}\zeta_{p}^t+\sum_{t\in D_1^{(p)}}\zeta_{p}^t=1,\\
&\sum_{t\in D_0^{(q)}}\zeta_{q}^t+\sum_{t\in D_1^{(q)}}\zeta_{q}^t=1,
\end{align*}
we know  there are at least $\frac{1}{2}\phi(2p^{m-a}q^{n-b})=\frac{(p-1)(q-1)}{2}p^{m-a-1}q^{n-b-1}$ many $k's$ satisfying $S'(\beta^k)\neq 0$.
Hence,  $S(\beta^{k})$ will have at most
\begin{align*}
&p^mq^n-1+\sum_{a=0}^m\sum_{b=0}^n\frac{(p-1)(q-1)p^{m-a-1}q^{n-b-1}}{2}\\
&=\frac{3p^mq^n-p^m-q^n-1}{2}.
\end{align*}
roots for $0\leq k\leq p^mq^n-1$. By Lemma \ref{lem1}, we obtain
\begin{align*}
LC(S)&\geq 2p^mq^n-\frac{3p^mq^n-p^m-q^n-1}{2}\\
&=\frac{(p^m+1)(q^n+1)}{2}.
\end{align*}
\end{remark}
\section*{Data Availability}
No data were used to support this study.
\section*{Conflicts of Interest}
The authors declare that they have no conflicts of interest.

\end{document}